\documentclass[letterpaper]{article}

\usepackage[preprint]{aaai2027}

\usepackage[hyphens]{url}
\usepackage{graphicx}
\usepackage{natbib}
\usepackage{caption}
\usepackage{algorithm}
\usepackage{algorithmic}
\usepackage{amsthm}
\usepackage{amssymb}
\usepackage{amsmath}
\usepackage{subcaption}
\usepackage{booktabs}

\usepackage{newfloat}
\usepackage{listings}
\DeclareCaptionStyle{ruled}{labelfont=normalfont,labelsep=colon,strut=off}
\floatstyle{ruled}
\newfloat{listing}{tb}{lst}{}
\floatname{listing}{Listing}

\usepackage{docmute}

\makeatletter
\let\arxiv@original@xfloat\@xfloat
\def\@xfloat#1[#2]{%
  \def\arxiv@float@opt{#2}%
  \def\arxiv@float@H{H}%
  \ifx\arxiv@float@opt\arxiv@float@H
    \arxiv@original@xfloat{#1}[t]%
  \else
    \arxiv@original@xfloat{#1}[#2]%
  \fi
}
\makeatother

\newtheorem{theorem}{Theorem}
\newtheorem{lemma}{Lemma}
\newtheorem{remark}{Remark}
\newtheorem{claim}{Claim}

\newcommand{\ALG}{\mathrm{ALG}}
\newcommand{\OPT}{\mathrm{OPT}}
\newcommand{\ADVICE}{\mathrm{ADVICE}}

\title{Learning-Augmented Algorithms for Online Vertex Cover}
\author{
Tianhang Lu,
Runtian Ren,
Shengcai Liu
}
\affiliations{
Guangdong Provincial Key Laboratory of Brain-Inspired Intelligent Computation,\\
Department of Computer Science and Engineering,\\
Southern University of Science and Technology, Shenzhen 518055, China\\
liusc3@sustech.edu.cn
}

\begin{document}

\let\arxivOriginalBibliography\bibliography
\renewcommand{\bibliography}[1]{}


\maketitle

\begin{abstract}
This paper studies learning-augmented online weighted vertex cover with local advice and a tradeoff parameter $\lambda \in (0,1)$.
We consider two graph settings: bipartite graphs and general graphs.
In both settings, the online algorithm must maintain a feasible vertex cover under irrevocable decisions.
We show that these problems admit the same robustness--consistency tradeoffs as learning-augmented ski rental.
For the bipartite graph model, we give a randomized algorithm that is $\frac{1}{1-e^{-\lambda}}$-robust and $\frac{\lambda}{1-e^{-\lambda}}$-consistent.
For the general graph model, we give a deterministic algorithm that is $(1+\frac{1}{\lambda})$-robust and $(1+\lambda)$-consistent.
We prove that the tradeoffs above are optimal in both settings.
We also validate the proposed algorithms through experiments on synthetic and real-world datasets.
\end{abstract}

\section{Introduction}
Online weighted vertex cover is a fundamental online covering problem with applications in link monitoring~\cite{yigit2021breadth}, multiprocessor scheduling~\cite{epstein2016vertex}, and the deployment of monitoring devices in urban road networks~\cite{gusev2020vertex}.
In the vertex-arrival model, the graph is revealed over time.
When a vertex arrives, the algorithm observes its edges to previously revealed vertices, and must immediately maintain a feasible vertex cover for all revealed edges.
Selecting a vertex incurs its weight, and the decision is irrevocable.
The objective is to minimize the total weight of the selected vertices, relative to an optimal offline solution that knows the entire graph.

Worst-case competitive analysis provides guarantees against adversarial arrivals, but it can be overly conservative when useful predictive information is available.
Learning-augmented algorithms address this limitation by allowing algorithms to exploit predictions while guaranteeing the worst-case performance ratio when the predictions are unreliable~\cite{purohit2018improving}.
The two standard requirements in the learning-augmented setting are \textbf{robustness} and \textbf{consistency}: robustness evaluates the worst-case competitive ratio under arbitrary predictions, whereas consistency evaluates the improved performance guarantee when the predictive information is accurate.

Our advice model reflects how machine-learning predictions can be used in an online setting.
Specifically, at each arrival, the advice chooses between two canonical covering actions: selecting $v$ or selecting all vertices in its revealed neighborhood $N(v)$.
The advice bit $a_v$ specifies the preferred action, yielding a binary prediction interface.
Moreover, any sequence of advice bits induces a feasible vertex cover, so prediction errors affect only the solution cost, not feasibility.

Specifically, an algorithm is said to be $r$-robust if its expected cost is at most $r \cdot \OPT$ where $\OPT$ is the cost of an optimal offline vertex cover.
Similarly, an algorithm is said to be $c$-consistent if its expected cost is at most $c \cdot \ADVICE$ where $\ADVICE$ is the cost of the feasible vertex cover induced by the advice.
This benchmark is well aligned with the learning-augmented framework, where the prediction may itself specify a feasible solution rather than merely provide a scalar estimate of the future, as in~\cite{spaeh2023online,choo2026learning}.
In our setting, the advice directly defines a feasible covering policy, and therefore $\ADVICE$ is the natural reference solution for measuring consistency.
Moreover, this notion subsumes the perfect-advice benchmark: if the advice-induced cover coincides with an offline optimum, then consistency with respect to $\ADVICE$ reduces exactly to the usual $\OPT$-based consistency guarantee.

Online bipartite vertex cover generalizes ski rental~\cite{wang2015two}: a ski-rental instance can be represented by a star graph.
Selecting the offline center corresponds to buying skis, whereas selecting the arriving leaves corresponds to renting skis each day.
One previous work by Purohit et al.~\shortcite{purohit2018improving} studied ski rental in the learning-augmented setting with a tradeoff parameter $\lambda \in (0,1)$ and a prediction of the season length.
They proposed a deterministic algorithm with $(1+\frac{1}{\lambda})$-robustness and $(1+\lambda)$-consistency, and a randomized algorithm with $\frac{1}{1-e^{-\lambda}}$-robustness and $\frac{\lambda}{1-e^{-\lambda}}$-consistency.
Here, $\lambda \in (0,1)$ denotes a tradeoff parameter.
Later, Wei and Zhang~\shortcite{wei2020optimal} proved that these robustness--consistency tradeoffs are tight.
In ski rental, consistency is evaluated under perfect predictions and therefore against $\OPT$.
Under the star representation above, this corresponds to the special case in which the induced advice solution is optimal.

Online vertex cover also falls within the broader framework of online set covering.
Existing learning-augmented guarantees for general online covering problems are typically non-constant, often scaling logarithmically with parameters of the input.
For example, Bamas et al.~\shortcite{bamas2020primal} gave an $O(\log(d/\lambda))$-robust and $O(1/(1-\lambda))$-consistent algorithm for online set cover, where $d$ denotes the maximum number of sets that cover any element.
It is worth noting that the robustness and consistency guarantees of such algorithms are weaker than those achieved for the ski rental problem.
As a result, this raises the following natural question:
\begin{quote}
\textit{Can the robustness and consistency known for ski rental be achieved for online vertex cover?}
\end{quote}

\paragraph{Our Contributions.} We answer this question affirmatively for two online vertex cover models.
In the bipartite setting, offline vertices are known in advance, online vertices arrive one by one, and every revealed edge connects the arriving vertex to an offline vertex.
In the general setting, the offline side may be empty, and each arriving vertex may be adjacent to either offline vertices or previously revealed online vertices.
The contributions of this work are as follows.
\begin{enumerate}
    \item We propose a randomized learning-augmented online algorithm for vertex cover on bipartite graphs, with $\frac{1}{1-e^{-\lambda}}$-robustness and $\frac{\lambda}{1-e^{-\lambda}}$-consistency.
    \item We propose a deterministic learning-augmented online algorithm for vertex cover on general graphs, with $(1+\frac{1}{\lambda})$-robustness and $(1+\lambda)$-consistency.
    \item We prove that the tradeoffs above are optimal for randomized algorithms in the bipartite setting and for deterministic algorithms in the general graph setting.

    \item We conduct numerical experiments to evaluate the performance of our algorithms on both synthetic datasets and real-world datasets.
    The results demonstrate that the proposed algorithms outperform existing learning-augmented algorithms and achieve performance comparable to that of the algorithm with the optimal competitive ratio in most cases.
\end{enumerate}

\paragraph{Our Techniques.} For the bipartite graph model, we extend the water-filling algorithm of Wang and Wong~\shortcite{wang2015two} with an advice-dependent threshold function.
The threshold increases with the water level but decreases when the arriving vertex is recommended, so the algorithm invests less in its neighbors.
Otherwise, it raises the threshold and invests more in them. This design balances robustness with consistency.
For the general-graph model, the challenge is that a previously revealed vertex may later become recommended through the advice of a newly arriving neighbor.
Building on the primal-dual framework, we assign each vertex a dynamic trust threshold.
When a later advice bit recommends selecting an already revealed vertex, its threshold is lowered accordingly. This allows the algorithm to adapt to evolving advice while preserving the desired guarantees.

\section{Related Work}\label{sec:related-work}
The theory of learning-augmented algorithms incorporates unreliable advice into algorithm design and characterizes how algorithmic performance varies with the quality of the advice.
It has been widely studied for many fundamental problems (see the survey~\cite{mitzenmacher2022algorithms}), such as paging~\cite{lykouris2021competitive}, matching~\cite{antoniadis2020secretary,choo2026learning,burathep2026learning}, and set cover~\cite{bamas2020primal,grigorescu2022learning}.
We next review prior work on learning-augmented ski rental, online covering, and online vertex cover.

\paragraph{Learning-Augmented Ski Rental Problem.}
The ski rental problem studies how an online decision maker should trade off continuing to rent against purchasing outright when the usage duration is unknown.
It can be viewed as a special case of online vertex cover~\cite{wang2015two}.
Under the classical competitive-analysis framework, the optimal competitive ratio is $2$ for deterministic algorithms~\cite{boyar2022relaxing} and $e/(e-1)$ for randomized algorithms~\cite{karlin1994competitive}.
However, such worst-case analysis does not exploit the historical data or machine-learning predictions that are often available in practical settings.
To overcome this limitation, Purohit et al.~\shortcite{purohit2018improving} introduced machine-learned predictions into the ski rental problem, following the basic principle of learning-augmented algorithms: an algorithm achieves improved performance when the predictions are accurate, while maintaining controlled robustness (i.e., competitiveness) when the predictions are erroneous.
Wei and Zhang~\shortcite{wei2020optimal} later established tight lower bounds for the consistency--robustness tradeoff in classical learning-augmented ski rental.
Other previous works further extended this line of research to settings with multiple expert predictions~\cite{gollapudi2019online}, multi-shop ski rental~\cite{wang2020online}, and variants such as multi-option ski rental~\cite{shin2023improved} and two-level ski rental~\cite{zhang2024learning}.

\paragraph{Learning-Augmented Online Covering Problem.}
A number of recent works have studied learning-augmented algorithms for online set cover and related online covering problems.
One common setting assumes that the prediction itself constitutes a feasible solution, together with a tradeoff parameter $\lambda \in (0,1)$.
Bamas et al.~\shortcite{bamas2020primal} developed a primal-dual framework for the online set cover problem and obtained an $O(\log(d/\lambda))$-robust and $O(1/(1-\lambda))$-consistent algorithm, where $d$ denotes the maximum number of sets that cover one element.
Grigorescu et al.~\shortcite{grigorescu2022learning} studied online covering LPs and SDPs with fractional predictions.
Under a relaxation that allows a polynomial number of constraints to be violated, they obtained an $O(\log(n/\lambda))$-robust and $O(1/(1-\lambda))$-consistent algorithm, where $n$ is the number of variables.
Subsequent work by Grigorescu et al.~\shortcite{grigorescu25a} further extended this line of research to more general convex objective functions.
Learning-augmented algorithms have also been investigated in settings where the prediction is not necessarily a feasible solution.
For instance, Ameli et al.~\shortcite{jabal2025learning} studied the case where the prediction consists of the set of requests that will appear in the future, while Nguyen~\shortcite{nguyen2026online} considered online covering problems with multiple experts.

\paragraph{Online Vertex Cover.}
The study of online vertex cover dates back to the work of Demange and Paschos~\shortcite{demange2005line}, who systematically analyzed vertex-by-vertex and cluster-revealing models.
Subsequently, Birmelé, Delbot, and Laforest~\shortcite{birmele2009mean} studied online vertex cover algorithms from the perspective of average-case performance, thereby enriching the understanding of the problem beyond worst-case analysis.
Wang and Wong~\shortcite{wang2015two} gave an algorithm with the optimal competitive ratio of $e/(e-1)$ for bipartite graphs, and obtained a $1.901$-competitive algorithm for the fractional version in general graphs.
Since then, further variants of online vertex cover have also been investigated, including online 3-path vertex cover~\cite{zhang2020algorithm} and online vertex cover with submodular or matroid constraints~\cite{wang2016matroid}.

\section{Preliminaries}\label{sec:preliminaries}
\paragraph{Bipartite and General Graphs.}
Let $G = (U \cup V,E)$ denote a bipartite graph with a nonnegative weight $w_v \ge 0$ for each vertex $v \in U \cup V$.
The vertices in $U$ are offline vertices, with weights known to the algorithm in advance.
The vertices in $V$ are online vertices that arrive one at a time.
Upon the arrival of an online vertex $v \in V$, its neighborhood $N(v) = \{u \in U \mid (u,v) \in E\}$ and its weight $w_v$ are both revealed.
The algorithm needs to make irrevocable decisions immediately at each vertex arrival, while ensuring that every edge revealed thus far is covered.
The objective is to minimize the total weight of the selected vertices.

Unlike in the bipartite setting, the offline vertex set $U$ may be empty in the general-graph model, and edges may connect two online vertices.
Upon the arrival of $v\in V$, the algorithm observes
$N(v)=\{u\in U\cup T\mid (u,v)\in E\}$,
where $T$ is the set of previously arrived online vertices.


\paragraph{Advice.}
When an online vertex $v \in V$ arrives, the algorithm receives a bit $a_v \in \{0,1\}$.
If $a_v = 1$, the advice recommends covering the edges incident to $v$ by selecting $v$ itself.
If $a_v = 0$, the advice recommends covering the edges incident to $v$ by selecting all neighbors of $v$.
The vertex set constructed by following the advice upon each arrival is always a feasible vertex cover, which we call the \textit{advice-induced vertex cover}.
Thus the advice-induced vertex cover is given by
\begin{equation}\label{eq:adv-cover}
\{v\in V\mid a_v=1\}\cup\bigcup_{v\in V:\,a_v=0} N(v).
\end{equation}
Note that if the vertices assigned bit $1$ themselves form a feasible vertex cover, then the advice-induced vertex cover is
$\{v\in V\mid a_v=1\}$.

\paragraph{Performance Measures.}
For an instance $I$, let $\ALG(I)$ denote the cost of a deterministic algorithm and let $\mathbb{E}[\ALG(I)]$ denote the expected cost of a randomized algorithm.
A randomized algorithm is $r$-robust if $\mathbb{E}[\ALG(I)]\le r\OPT(I)$ for every instance $I$ and every advice sequence.
It is $c$-consistent if $\mathbb{E}[\ALG(I)]\le c\ADVICE(I)$ for every instance $I$ and every advice sequence.
For a deterministic algorithm, the same definitions apply without the expectation.

\section{Online Bipartite Vertex Cover with Advice}\label{sec:bi-vc}
In this section, we present a learning-augmented algorithm for online bipartite vertex cover.
Our learning-augmented online algorithm achieves both robustness and consistency.
We first introduce the algorithm, and then prove its robustness and consistency.

\paragraph{Algorithm Description.}
Given a tradeoff parameter $\lambda\in(0,1)$, define the threshold function $f:[0,1]\times\{0,1\}\to\mathbb{R}$ by
$$
f(y,a)=
\begin{cases}
y+\alpha(\lambda), & a=0,\\
y+\beta(\lambda), & a=1,
\end{cases}
$$
where $\alpha(\lambda)=1/(1-e^{-\lambda})-1$, and $\beta(\lambda)=\lambda/(1-e^{-\lambda})-1$.
The main idea of our algorithm is to maintain a fractional value $y_u$ for every revealed vertex $u$, interpreted as the potential already placed on $u$.
When an online vertex $v$ arrives, the algorithm updates the potentials of its revealed neighborhood $N(v)$ by a water-filling rule.
For a target level $y$, the quantity
$$
\sum_{u\in N(v)} w_u\max\{y-y_u,0\}
$$
is the total additional weighted potential needed to raise all neighbors below level $y$ up to $y$.

The function $f(y,a_v)$ determines how much potential the algorithm is willing to spend on $N(v)$ at level $y$, with the advice $a_v$ controlling this investment.
For fixed advice, $f(y,a_v)$ increases with $y$, allowing more potential when the target level is higher.
Since $\alpha(\lambda)>\beta(\lambda)$, the threshold is larger when $a_v=0$ and smaller when $a_v=1$.
Thus, if the advice suggests selecting the neighbors of $v$, the algorithm invests more in $N(v)$ and may choose a larger level $y$.
If the advice suggests selecting $v$, it invests less in the neighbors and chooses a smaller feasible level.

The algorithm chooses the largest feasible level $y\leq 1$, raises each neighbor $u\in N(v)$ to $\max\{y_u,y\}$, and assigns the remaining potential to the arriving vertex by setting $y_v=1-y$.
Hence every new edge is fractionally covered upon the arrival of $v$, and previously covered edges remain covered because potentials only increase.

Finally, the algorithm converts the fractional solution into an integral vertex cover using one random threshold $t \sim \mathrm{Unif}[0,1]$ sampled between 0 and 1 at the beginning.
After computing $y$ for an arriving vertex $v$, it selects $v$ if $y \le t$.
Otherwise, the algorithm selects all vertices in $N(v)$.
In this way, every newly revealed edge $(u,v)$ is guaranteed to be covered.
For completeness, the pseudocode is given in Algorithm~\ref{alg:bi-vc}.
\begin{algorithm}[tbp]\small
    \caption{Online Bipartite Vertex Cover with Advice}
	\begin{algorithmic}[1]
        \STATE {\bfseries Input:} Online graph $G=(U\cup V,E)$, weights $\{w_v\}_{v\in U\cup V}$, advice $a$, parameter $\lambda$
        \STATE {\bfseries Output:} A vertex cover $C$
        \STATE Sample $t$ uniformly from $[0,1]$
        \STATE $C\leftarrow\emptyset$
        \STATE Initialize $y_u=0$ for each $u\in U$
        \FOR {each online vertex $v$ and its advice $a_v$}
            \STATE Maximize $y\leq 1$ subject to
            $$
            \sum_{u\in N(v)}w_u\max\{y-y_u,0\}\leq w_vf(y,a_v)
            $$ \label{step:max-y}
            \STATE For each $u\in N(v)$, set $y_u\leftarrow \max\{y_u,y\}$\label{step:potential-L}
            \STATE Set $y_v\leftarrow 1-y$\label{step:potential-R}
            \IF{$y\leq t$}
                \STATE $C\leftarrow C\cup\{v\}$\label{step:round-R}
            \ELSE
                \STATE $C\leftarrow C\cup N(v)$\label{step:round-L}
            \ENDIF
        \ENDFOR
        \RETURN $C$
	\end{algorithmic}
	\label{alg:bi-vc}
\end{algorithm}

\paragraph{Robustness and Consistency Analysis.}
We now establish the robustness and consistency guarantees of the algorithm.
For any vertex subset $S \subseteq U \cup V$, define the notation $w(S) := \sum_{v \in S} w_v$.
We next establish two lemmas, whose proofs are deferred to the supplementary material.
We first show that the randomized rounding preserves the weighted fractional cost in expectation.

\begin{lemma}\label{lem:bi-vc-exp-cost}
The online vertex cover selected by Algorithm~\ref{alg:bi-vc} has an expected cost of $\sum_{v\in U\cup V}w_v y_v$.
\end{lemma}

We next show that the level chosen in Step~\ref{step:max-y} either reaches $1$ or makes the constraint tight, as expected from the water-filling stopping rule.

\begin{lemma}\label{lem:com-y}
Let $y$ be an optimal solution in Step~\ref{step:max-y}.
Then either $y=1$ or $\sum_{u\in N(v)}w_u\max\{y-y_u,0\}=w_vf(y,a_v)$.
\end{lemma}
Thanks to Lemma~\ref{lem:bi-vc-exp-cost} and~\ref{lem:com-y}, we can establish the robustness and consistency for Algorithm~\ref{alg:bi-vc}.
\begin{theorem}\label{thm:bigraph-alg-robust-cnst}
Given any tradeoff parameter $\lambda \in (0,1)$, Algorithm~\ref{alg:bi-vc} is $\frac{1}{1-e^{-\lambda}}$-robust and $\frac{\lambda}{1-e^{-\lambda}}$-consistent for online bipartite vertex cover.
\end{theorem}

\begin{proof}
Let $\{y_v\}_{v\in U\cup V}$ denote the fractional vertex cover maintained by the algorithm.
By Lemma~\ref{lem:bi-vc-exp-cost}, we only need to upper bound the weighted fractional value $\sum_v w_v y_v$.
Given any vertex cover $S$ of $G$, we charge the total weighted fractional value generated by Algorithm~\ref{alg:bi-vc} to vertices in $S$.
Consider the arrival of an online vertex $v \in V$ with advice $a_v$.
Let $x_u$ be the value of $u \in N(v)$ before the update, and let $y$ be the level chosen by Algorithm~\ref{alg:bi-vc}.
Define
$$
d_u=(y-x_u)^+,\qquad
D_v=\sum_{u\in N(v)} w_u d_u .
$$
Then Algorithm~\ref{alg:bi-vc} increases the weighted fractional value of offline vertices by $D_v$ and assigns weighted value $w_vy_v=w_v(1-y)$ to $v$.
Denoting by $c(a,\lambda) = (1-a) \cdot \alpha(\lambda) + a \cdot \beta(\lambda)$, we have $f(y,a) = y + c(a,\lambda)$.

If $v \in S$, we charge both $w_vy_v$ and all weighted increments $\{w_ud_u\}_{u\in N(v)}$ to $v$.
Since $D_v\le w_v \cdot f(y,a_v)$ (see Step~\ref{step:max-y}), the total charge to $v$ is at most
\begin{equation}
\begin{aligned}\label{neq:pot-online}
w_v \cdot (1-y) + D_v&\le w_v \cdot (1-y+f(y,a_v))\\
&= w_v \cdot (1+c(a_v,\lambda)).
\end{aligned}
\end{equation}

If $v \notin S$, since $S$ is a feasible vertex cover, we have $N(v) \subseteq S$.
We charge each weighted increment $w_ud_u$ to $u\in N(v)$.
We also distribute the weighted value $w_v(1-y)$ of the online vertex among the vertices in $N(v)$ proportionally to their weighted increments.
Namely, for each $u\in N(v)$ we charge an additional amount
$$
w_u \cdot \frac{1-y}{f(y,a_v)} \cdot d_u
$$
to $u$.
By Lemma~\ref{lem:com-y}, we have either $y=1$ or $D_v=w_vf(y,a_v)$.
If $D_v = w_v f(y,a_v)$, these additional charges sum to $w_v(1-y)=w_vy_v$.
If $y=1$, then $y_v=0$, and there is no online value to charge.
Thus every unit of weighted fractional value is charged to some vertex in $S$.

It remains to upper bound the charge received by each vertex in $S$ relative to its weight.
For $a\in\{0,1\}$, the function $(1-t)/(t+c(a,\lambda))$ is non-increasing in $t$.
Hence, whenever an offline vertex $u$ is raised from level $x$ to level $y$, the additional charge assigned to $u$ due to the online vertex satisfies
$$
w_u \cdot \frac{1-y}{f(y,a)} \cdot (y-x)\le w_u\int_x^y \frac{1-t}{f(t,a)}\,dt.
$$
Moreover, the total direct weighted increments charged to $u$ are at most $w_u$, since the level of $u$ never exceeds $1$.
Fix an offline vertex $u\in U\cap S$.
Consider only the arrivals of online vertices $v\notin S$ that raise the value of $u$.
Let these vertices be $v_1,\ldots,v_k$ in chronological order, and let $a_i$ be the advice bit of $v_i$.
Let $x_i$ and $x_{i+1}$ denote the value of $u$ immediately before and after processing $v_i$ respectively.
The total charge to $u$ is at most
\begin{equation}\label{neq:pot-offline}
w_u\left(1+\sum_{i=1}^k \int_{x_i}^{x_{i+1}}\frac{1-t}{t+c(a_i,\lambda)}\,dt\right).
\end{equation}
\begin{claim}\label{clm:bi-vc-bound}
Given any feasible vertex cover $S$, the expectation of total weight of the vertex cover selected by Algorithm~\ref{alg:bi-vc} is at most $\sum_{u \in S}$~(\ref{neq:pot-offline}).
\end{claim}

Using Claim~\ref{clm:bi-vc-bound}, we derive the two guarantees on robustness and consistency.
For robustness guarantee, let $S$ denote a minimum-weight vertex cover in $\OPT$.
Since $\beta(\lambda)\le \alpha(\lambda)$, every online vertex $v \in S$ receives charge at most $w_v(1+\alpha(\lambda))$ by~(\ref{neq:pot-online}).
For an offline vertex $u \in U \cap S$, $c(a,\lambda)\geq \beta(\lambda)$.
Thus, applying inequality~(\ref{neq:pot-offline}) with $c(a,\lambda)$ replaced by its lower bound $\beta(\lambda)$, the total charge to $u$ is at most
\begin{align*}
w_u\left(1+\int_0^1 \frac{1-t}{t+\beta(\lambda)}\,dt\right)
&\le w_u(1+\alpha(\lambda)).
\end{align*}
Hence
$$
\sum_{z\in U\cup V} w_z y_z
\le (1+\alpha(\lambda)) \cdot w(S)
= (1+\alpha(\lambda)) \cdot \OPT,
$$
which proves $1/(1-e^{-\lambda})$-robustness.

For consistency guarantee, let $S$ denote the advice-induced vertex cover.
If an online vertex $v\in V$ belongs to $S$, then its advice is $a_v=1$, and the charge to $v$ is at most $w_v(1+\beta(\lambda))$ by inequality~(\ref{neq:pot-online}).
If $v\notin S$, then $a_v=0$.
Hence, for every offline vertex $u\in U\cap S$, all additional distributed charges come only from steps with $c(a_v,\lambda)=\alpha(\lambda)$.
Therefore, by applying inequality~(\ref{neq:pot-offline}), the total charge to $u$ is at most
$$
w_u\left(1+\int_0^1 \frac{1-t}{t+\alpha(\lambda)}\,dt\right)
= w_u \cdot (1+\beta(\lambda)).
$$
Hence
$$
\sum_{z\in U\cup V} w_z y_z
\le (1+\beta(\lambda)) \cdot w(S)
= (1+\beta(\lambda)) \cdot \ADVICE,
$$
which proves $\lambda/(1-e^{-\lambda})$-consistency.
\end{proof}

\begin{remark}
The parameter $\lambda$ controls the consistency--robustness tradeoff.
As $\lambda\to 0$, the consistency ratio $\lambda/(1-e^{-\lambda})$ tends to $1$, while the robustness ratio $1/(1-e^{-\lambda})$ becomes unbounded.
As $\lambda\to 1$, robustness approaches $e/(e-1)$, and the algorithm reduces to a random $e/(e-1)$-competitive online bipartite vertex cover algorithm, matching the known lower bound~\cite{wang2015two}.
\end{remark}

\section{Online Vertex Cover with Advice}\label{sec:vc}
In this section, we present our deterministic algorithm for online vertex cover in general graphs with advice.
The general-graph model extends the bipartite model by allowing edges between online vertices.
In particular, when $U=\emptyset$, the graph is induced entirely by the online vertex set $V$.
In the following, we first describe the proposed algorithm, with its pseudocode given in Algorithm~\ref{alg:vc}, and then analyze its robustness and consistency. 

\paragraph{Algorithm Description.}
For each offline or already revealed vertex $z$, the algorithm maintains a load $\ell(z)$ and a threshold $\tau(z)$.
Initially, for every offline vertex $u \in U$, $\ell(u) = 0$, and $\tau(u) = w_u / \lambda$.
Upon the arrival of an online vertex $v$, $\ell(v) = 0$, and
$$
\tau(v) = \begin{cases}
w_v / \lambda, & a_v = 0,\\
w_v, & a_v = 1.
\end{cases}
$$
The threshold encodes the advice-induced preference as selection difficulty: advice-recommended vertices are assigned lower thresholds and hence easier to select, whereas non-recommended vertices require more accumulated load before being selected.
The parameter $\lambda$ controls this bias: smaller $\lambda$ increases reliance on the advice, while larger $\lambda$ weakens it and improves robustness.

Whenever a newly revealed edge $(u,v)$ is still uncovered, the algorithm increases the loads of both endpoints at the same rate until one endpoint reaches its current threshold.
Every endpoint whose load reaches its threshold is added to the online vertex cover, which covers the edge.
\begin{algorithm}[tbp]\small
    \caption{Online Vertex Cover with Advice}
    \label{alg:vc}
	\begin{algorithmic}[1]
        \STATE {\bfseries Input:} Online graph $G=(U\cup V,E)$, weights $\{w_z\}_{z\in U\cup V}$, advice $a$, parameter $\lambda$
        \STATE {\bfseries Output:} A vertex cover $C$
        \STATE $C\leftarrow \emptyset$
        \FOR{each offline vertex $u\in U$}
            \STATE $\ell(u)\leftarrow 0$, $\tau(u)\leftarrow w_u/\lambda$
        \ENDFOR
        \FOR{each online vertex $v\in V$ and its advice $a_v$}
            \STATE $\ell(v)\leftarrow 0$, $\tau(v)\leftarrow (1-a_v) \cdot w_v/\lambda + a_v \cdot w_v$
            \FOR{each newly revealed edge $(u,v)$}
                \IF{$a_v=0$}
                    \STATE $\tau(u)\leftarrow w_u$ \label{step:lower-threshold}
                    \IF{$u\notin C$ and $\ell(u)\ge \tau(u)$}
                        \STATE $C\leftarrow C\cup\{u\}$ \label{step:add-lowered-neighbor}
                    \ENDIF
                \ENDIF
                \IF{$u,v\notin C$}
                    \STATE $\delta\leftarrow \min\{\tau(u)-\ell(u),\tau(v)-\ell(v)\}$\label{step:load}
                    \STATE $\ell(u),\ell(v)\leftarrow \ell(u)+\delta,\ell(v)+\delta$ \label{step:incre-ell}
                    \STATE Add every $z\in\{u,v\}$ with $\ell(z)=\tau(z)$ to $C$ \label{step:add-to-C}
                \ENDIF
            \ENDFOR
        \ENDFOR
        \RETURN $C$
	\end{algorithmic}
\end{algorithm}

\paragraph{Robustness and Consistency Analysis.}
We now establish the robustness and consistency of the algorithm.
\begin{theorem}\label{thm:graph-alg-robust-cnst}
Given any tradeoff parameter $\lambda\in (0,1)$, Algorithm~\ref{alg:vc} is $(1+\frac{1}{\lambda})$-robust and $(1+\lambda)$-consistent for online vertex cover on general graphs.
\end{theorem}

\begin{proof}
The feasibility of the online vertex cover produced by Algorithm~\ref{alg:vc} follows immediately from Step~\ref{step:load} to~\ref{step:add-to-C}.
Each newly revealed uncovered edge triggers load increases until one of its endpoints is added to $C$, so the edge becomes covered.
Since $C$ is monotone increasing, previously covered edges remain covered.

Next we establish a charging inequality used for both robustness and consistency guarantees.
Let $S$ denote any feasible vertex cover of $G$ and set $B = C \setminus S$.
For every load increment $\delta$ received by a vertex $x \in B$ while processing an edge $(x,y)$, we charge $\delta$ to the other endpoint $y$.
This is well-defined since $S$ covers $(x,y)$ and $x \notin S$, hence $y \in S$.
Let $q_y^S$ denote the total charge received by $y \in S$.
Then every load increment of every vertex in $B$ is charged exactly once, and therefore
\begin{equation}\label{eq:graph-vc-general-charge}
    \sum_{x\in B}\ell(x) = \sum_{y \in S} q_y^S.
\end{equation}
Moreover, by Steps~\ref{step:add-lowered-neighbor} and~\ref{step:add-to-C} of Algorithm~\ref{alg:vc}, we have $\ell(z)\geq w_z$ for any $z \in C$.
Hence
\begin{align}
w(C) &\le w(S) + w(B)\\
&\le w(S) + \sum_{x \in B} \ell(x)\\
&=w(S) + \sum_{y \in S} q_y^S.\label{eq:graph-vc-general-bound}
\end{align}

For robustness guarantee, given any $y \in S$, the charge $q_y^S$ consists only of load increments also received by $y$.
Thus $q_y^S\le \ell(y)$.
Since loads never exceed the largest possible threshold, we have $\ell(y)\le w_y/\lambda$.
Applying \eqref{eq:graph-vc-general-bound} gives
$$
w(C)\le w(S)+\frac{1}{\lambda}w(S)=\left(1+\frac{1}{\lambda}\right)w(S).
$$
Taking $S$ to be an optimal offline vertex cover $\OPT$ proves $\left(1+\frac{1}{\lambda}\right)$-robustness.

For consistency guarantee, let $\widehat C$ denote the advice-induced vertex cover which is given by~(\ref{eq:adv-cover}).
By definition, $w(\widehat C)=\ADVICE$.
Set $S=\widehat C$ and $B=C\setminus \widehat C$.

Note that every vertex $x\in B$ is selected at threshold $w_x/\lambda$.
If $x\in U$, then if some neighbor $v\in V$ of $x$ arrives with $a_v=0$, we would have $x\in \widehat C$, a contradiction.
Therefore, the threshold of $x$ is never lowered to $w_x$.
If $x\in V$, then $a_x=0$ because all online vertices with advice bit $1$ belong to $\widehat C$.
Moreover, no later online vertex $v$ with $a_v=0$ reveals $x$ as a neighbor.
Otherwise $x\in\widehat C$.
Thus the threshold of $x$ is never lowered to $w_x$.
Consequently,
\begin{equation}\label{eq:graph-vc-advice-high-threshold}
w_x=\lambda \cdot  \ell(x),\qquad \forall \, x\in B .
\end{equation}

It remains to bound the charges received by vertices in $\widehat C$.
Fix $y\in \widehat C$, and consider any charge to $y$ generated by an
edge $(x,y)$ with $x\in B=C\setminus \widehat C$. We claim that such a
charge can occur only when $y$ is unselected and its current threshold
is $w_y$.

Indeed, if the edge is revealed when $x$ arrives, then $x\notin \widehat C$ implies $a_x=0$.
Hence, before processing this edge, the algorithm sets the threshold of the already revealed endpoint $y$ to $w_y$.
If this makes $\ell(y)\ge w_y$, then $y$ is immediately selected and the edge is already covered, so no charge is generated.

If the edge is revealed when $y$ arrives, then we must have $a_y=1$;
otherwise $a_y=0$ would imply $x\in N(y)\subseteq \widehat C$, contradicting
$x\in B$. Thus $y$'s threshold is $w_y$ from its arrival. In either
case, every actual charge to $y$ occurs while $y\notin C$ and
$\tau(y)=w_y$.

Each such charge is also a load increment of $y$. Since the algorithm
selects $y$ as soon as its load reaches its current threshold $w_y$,
the total charge received by $y$ is at most $w_y$. Therefore,
\begin{equation}\label{eq:graph-vc-advice-charge-bound}
q_y^{\widehat C} \le w_y,\qquad \forall \, y \in \widehat C .
\end{equation}

Combining \eqref{eq:graph-vc-general-charge}, \eqref{eq:graph-vc-advice-high-threshold}, and
\eqref{eq:graph-vc-advice-charge-bound}, we obtain
$$
w(B)= \lambda\sum_{x\in B}\ell(x)= \lambda\sum_{y\in\widehat C}q_y^{\widehat C}\le \lambda \cdot w(\widehat C).
$$
Thus we have
$$
w(C)\le w(\widehat C)+w(B)\le (1+\lambda) \cdot w(\widehat C)= (1+\lambda) \cdot \ADVICE.
$$
This proves $(1+\lambda)$-consistency and hence Theorem~\ref{thm:graph-alg-robust-cnst}.
\end{proof}

\begin{remark}
The parameter $\lambda$ controls the consistency--robustness tradeoff.
As $\lambda\to 0$, the consistency ratio $1+\lambda$ tends to $1$, while robustness $1+1/\lambda$ becomes unbounded.
As $\lambda\to 1$, robustness tends to $2$; both thresholds equal $w_x$, so the algorithm ignores the advice and becomes the standard deterministic primal-dual $2$-competitive algorithm, matching the known lower bound~\cite{boyar2022relaxing}.
\end{remark}

\section{Lower Bounds for Robustness--Consistency Tradeoffs}\label{sec:tradeoff}
In this section, we show that our algorithms achieve the best possible tradeoff between robustness and consistency.
These two lower-bound proofs are based on a specially constructed star instance, where the advice recommends selecting the center.
This instance is exactly a ski rental instance: selecting leaves corresponds to renting, while selecting the center corresponds to buying.
Robustness on short prefixes prevents the algorithm from buying too early, whereas robustness on long prefixes forces it to buy eventually.
Together, these constraints yield the classical ski rental lower bound and the stated robustness-consistency tradeoff.
Detailed proofs are deferred to the supplementary material.

\begin{theorem}\label{thm:bigraph-lower-bound}
If a randomized learning-augmented online algorithm for online bipartite vertex cover is $\gamma$-robust, then $\gamma \geq e/(e-1)$.
Moreover, any consistency ratio $c$ achieved by such an algorithm must satisfy
$$
c\geq \gamma\log(1+\frac{1}{\gamma-1}).
$$
In particular, for any $\lambda\in (0,1)$, any randomized algorithm achieving robustness $\gamma\leq 1/(1-e^{-\lambda})$ must have consistency $c\geq \lambda/(1-e^{-\lambda})$.
\end{theorem}

\begin{theorem}\label{thm:graph-lower-bound}
For any $\lambda\in (0,1)$, if a deterministic learning-augmented online algorithm for online vertex cover on general graphs is $(1+\frac{1}{\lambda})$-robust, then any consistency ratio $c$ achieved by such an algorithm must satisfy
$
c\geq 1+\lambda.
$
\end{theorem}

\begin{remark}
The ski rental tradeoff of~\cite{wei2020optimal} cannot be applied directly to prove Theorems~\ref{thm:bigraph-lower-bound} and~\ref{thm:graph-lower-bound}.
In classical ski rental, the prediction is a scalar estimate of the season length, and consistency is measured against $\OPT$ under perfect prediction.
In our model, however, the prediction is a sequence of local advice bits that induces a feasible vertex cover, and consistency is measured against the advice-induced solution $\ADVICE$.
Thus, our results can be viewed as extending the ski rental robustness-consistency tradeoff to this more general $\ADVICE$-based notion of consistency.
\end{remark}

\section{Experiments}\label{experiment}
In this section, we conduct numerical experiments on both synthetic and real-world datasets to evaluate the performance of Algorithm~\ref{alg:bi-vc} (LA-B) and Algorithm~\ref{alg:vc} (LA-G).
The source code is included in the supplementary archive.

\paragraph{Computational Settings.}
All experiments were conducted on a machine running Windows 11 with an AMD Ryzen 7 5800H CPU and 16 GB of memory.
We ran 100 independent random trials for each experimental setting.

\paragraph{Datasets.}
We evaluate the proposed algorithms on four groups of experiments, covering synthetic and real-world datasets under both bipartite and general-graph models.
For each graph, we randomly designate half of the vertices as offline vertices and the remaining half as online vertices.
The online vertices then arrive in a uniformly random order.
The weight of each vertex is sampled independently from the uniform distribution on $[0,1]$.
For the synthetic datasets, we construct Erdős-Rényi (ER) graphs as follows.
Given the number of vertices $n=1000$ and the edge probability $p \in \{0.1,0.2,0.5\}$, the ER graph contains $500$ offline vertices and $500$ online vertices.
In the bipartite graph model, each online-offline vertex pair is connected independently with probability $p$.
In the general graph model, each online-online or online-offline vertex pair is connected independently with probability $p$.

We use the nine snapshots in the Oregon-1 dataset, which contain Autonomous Systems peering information inferred from Oregon Route Views between March 31 and May 26, 2001~\cite{leskovec2005graphs,leskovec2016snap}.
Each snapshot contains approximately $10^4$ vertices and $2.2\times10^4$ edges.
In the bipartite graph model, we remove all offline-offline and online-online edges, whereas in the general graph model, we remove all offline-offline edges.

\paragraph{Advice.}
We generate the advice as follows.
For each graph, we first compute an exact offline vertex cover $x$.
In the bipartite graph model, this optimum is obtained by the standard minimum-cut reduction for weighted bipartite vertex cover.
In the general graph model, it is obtained by solving the 0-1 vertex-cover integer program using scipy.optimize.milp of Python 3.9, backed by the HiGHS optimizer.
When an online vertex $v\in V$ arrives, we set $a_v$ to the flipped value of $x_v$, namely $a_v=1-x_v$, with probability $\eta$ as the replacement rate.
Otherwise, we set $a_v=x_v$.
The same generation procedure has also been used in~\cite{grigorescu2022learning}.

\paragraph{Benchmarked Algorithms.}
For the general graph model, we consider three baselines: the \textit{Blind-Following} algorithm, the well-known 2-competitive \textit{Primal-Dual} algorithm~\cite{buchbinder2009online}, and \textit{PDLA}~\cite{bamas2020primal}.
The Blind-Following algorithm follows the advice-induced vertex cover which is given by~(\ref{eq:adv-cover}).
For the bipartite graph model, we further include the \textit{GreedyAllocation} algorithm of Wang and Wong~\shortcite{wang2015two} as an additional baseline; it achieves the optimal competitive ratio of $e/(e-1)$.
Note that GreedyAllocation and the Primal-Dual algorithm do not rely on any predictions, and hence their performance remains unchanged for any replacement rate.
In addition, we consider three values of $\lambda$, namely $0.25$, $0.5$, and $0.75$, to investigate the impact of $\lambda$ on algorithmic performance.
\begin{figure}[t]
\centering
\includegraphics[width=\linewidth,height=0.18\textwidth]{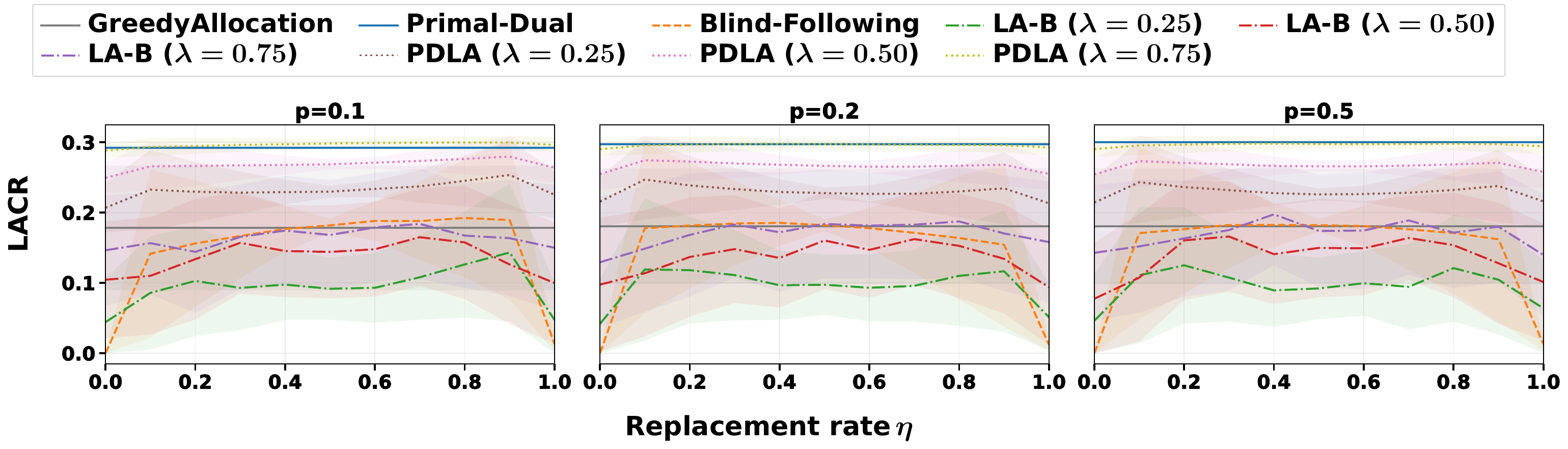}
\caption{LACR of the algorithms on synthetic ER datasets under the bipartite graph model with $p\in\{0.1,0.2,0.5\}$.}
\label{fig:er_bipartite}
\end{figure}

\paragraph{Results.}
The logarithmic average competitive ratio (LACR) is defined as the average, over all trials, of the logarithm of the online algorithm’s total cost divided by the optimal offline cost.
Since the instances are random, we plot the average over 100 trials for each setting of $\lambda$, with the shaded regions showing 1 standard deviation.
Figures~\ref{fig:er_bipartite} and~\ref{fig:er_general} report the LACR of the algorithms on synthetic ER graphs under the bipartite and general graph models respectively, for $p\in\{0.1,0.2,0.5\}$.
Figures~\ref{fig:real_bipartite} and~\ref{fig:real_general} report the corresponding results on the real-world graphs.

The experimental results show that, at low replacement rates, LA-B and LA-G achieve lower empirical competitive ratios than the baselines.
Their performance degrades smoothly as the replacement rate increases and remains comparable to that of the advice-free baselines at high replacement rates.
For each fixed value of $\lambda$ and each replacement rate, our algorithms also outperform PDLA.
\begin{figure}[t]
    \centering
    \includegraphics[width=\linewidth,height=0.18\textwidth]{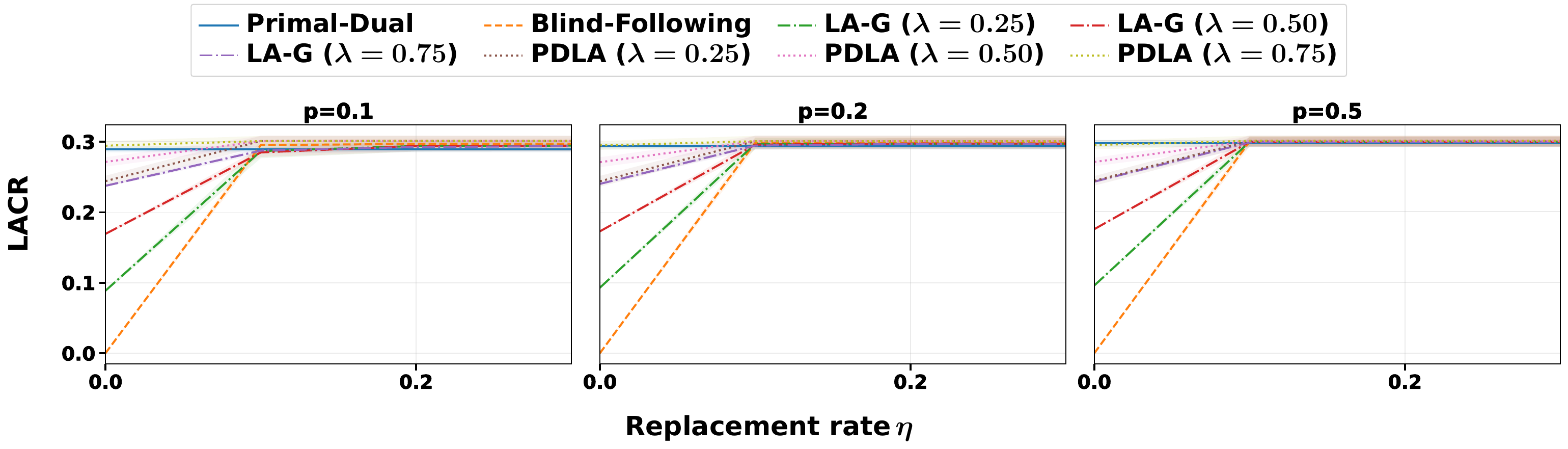}
    \caption{LACR of the algorithms on synthetic ER datasets under the general graph model with $p\in\{0.1,0.2,0.5\}$.}
    \label{fig:er_general}
\end{figure}

\begin{figure}[t]
    \centering
    \begin{subfigure}{0.23\textwidth}
        \centering
        \includegraphics[width=\linewidth,height=1.0\textwidth]{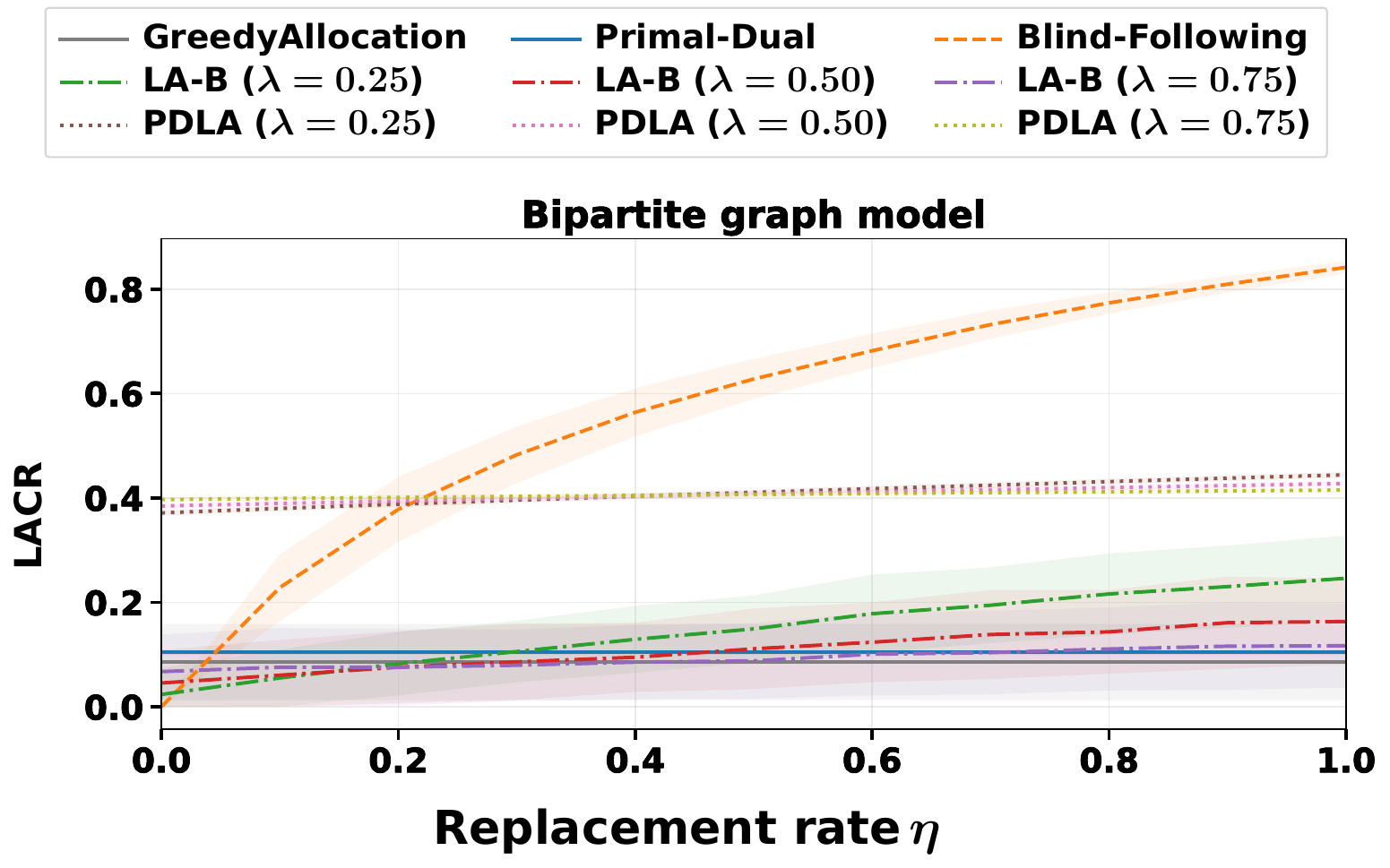}
        \caption{Bipartite}
        \label{fig:real_bipartite}
    \end{subfigure}
    \hfill
    \begin{subfigure}{0.23\textwidth}
        \centering
        \includegraphics[width=\linewidth,height=1.0\textwidth]{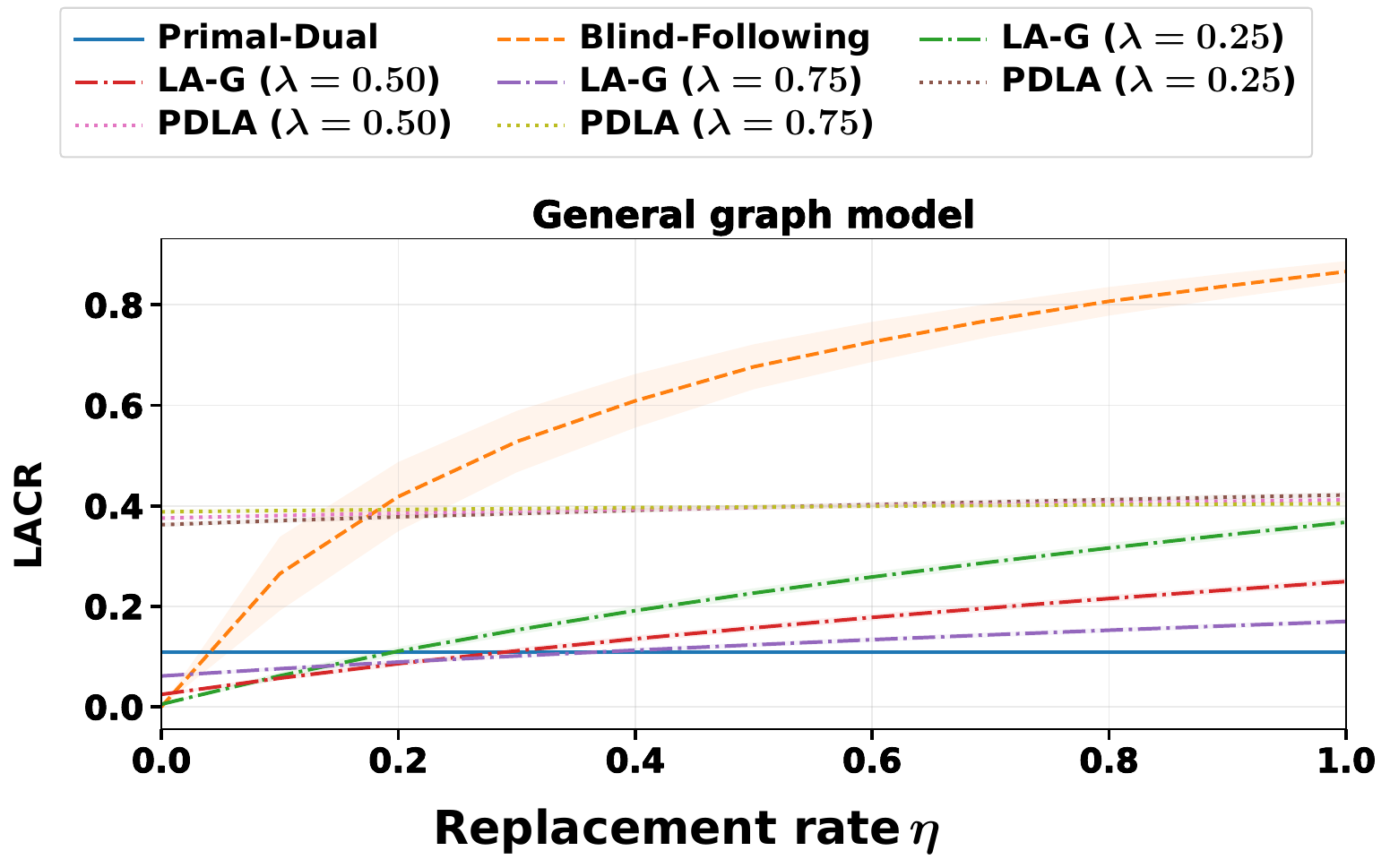}
        \caption{General}
        \label{fig:real_general}
    \end{subfigure}
    \caption{LACR of the algorithms on the Oregon-1 datasets: (a) bipartite graph model; (b) general graph model.}
    \label{fig:real_bipartite-general}
\end{figure}

\section{Conclusion}\label{conclusion}
Motivated by the success of learning-augmented algorithms for the ski rental problem, we study the online vertex cover problem on both bipartite graphs and general graphs. For these two models, we design learning-augmented algorithms that achieve both robustness and consistency.
Furthermore, we investigate the robustness-consistency tradeoffs of learning-augmented algorithms for online bipartite vertex cover and online vertex cover in general graphs.

Future research directions include designing learning-augmented randomized algorithms with improved robustness and consistency.
Moreover, extending the proposed model to nonlinear objective functions, such as submodular functions, is also worth further investigation.
In addition, designing learning-augmented algorithms that can adaptively select the parameter $\lambda$ to improve performance is also an important direction for future work.

\bibliography{aaai2027}


\clearpage
\appendix
\section*{Supplementary Material}

\setcounter{figure}{0}
\setcounter{table}{0}
\setcounter{equation}{0}
\setcounter{algorithm}{0}
\setcounter{listing}{0}
\renewcommand{\thefigure}{S\arabic{figure}}
\renewcommand{\thetable}{S\arabic{table}}
\renewcommand{\theequation}{S\arabic{equation}}
\renewcommand{\thealgorithm}{S\arabic{algorithm}}
\renewcommand{\thelisting}{S\arabic{listing}}

\let\arxivOriginalMaketitle\maketitle
\def\maketitle{}


\title{Supplementary Material}

\author{
Anonymous Submission
}
\affiliations{
}


\maketitle

\paragraph{Proof of Lemma 1. }
Let $\{y_v\}_{v\in U\cup V}$ denote the fractional vertex cover maintained by the algorithm.
For an online vertex $v\in V$, suppose that the algorithm chooses the level $y$ when $v$ arrives.
Since $y=0$ is always feasible, the maximizer satisfies $y\geq 0$.
Then $y_v=1-y$, and $v$ is selected if and only if $y\le t$.
Since $t$ is uniformly sampled from $[0,1]$, we have $\mathrm{Pr}[v\in C]=\mathrm{Pr}[t\ge y]=1-y=y_v$.
For an offline vertex $u\in U$, $u$ is selected if and only if $t<y_u$.
Hence $\mathrm{Pr}[u\in C]=y_u$.
Therefore, we have the following claim:
$$
\mathbb{E}[w(C)]
=\sum_{v\in U\cup V}w_v\mathrm{Pr}[v\in C]
=\sum_{v\in U\cup V}w_v y_v.
$$

\paragraph{Proof of Lemma 2.}
If $y < 1$, we have $D_v = w_v f(y,a_v)$; otherwise, $y$ could still be increased slightly, contradicting the maximality of $y$.
Hence these additional charges sum to $w_v(1-y)=w_vy_v$.
If $y=1$, then $y_v=0$, and there is no online value to charge.
Thus every unit of weighted fractional value is charged to some vertex in $S$.

\paragraph{Proof of Theorem 3. }
Let $A$ be a randomized $\gamma$-robust algorithm.
Consider a star with $U=\{u\}$ and online vertices $V$ with $|V|\geq 1$, where $w_u=1$ and every online vertex has weight $\epsilon\in(0,1)$.
The advice is $a_v=0$ for all $v\in V$, so the advice-induced vertex cover is $\{u\}$.
Run algorithm $A$ on the infinite sequence and let $K$ be the number of online vertices selected before $A$ first selects $u$, with $K=\infty$ if $u$ is never selected.
Define $F_r=\mathrm{Pr}[K\le r]$.
Let $\mathbb{E}[A_m]$ denote the expected cost incurred by algorithm $A$ when $|V|=m$.
\begin{align}
\mathbb{E}[A_m] &\geq\epsilon m\mathrm{Pr}[K\geq m]+\sum_{j=0}^{m-1}(j\epsilon+1)\mathrm{Pr}[K=j]\\
&=\epsilon\sum_{j=0}^{m-1}\mathrm{Pr}[K>j]+\mathrm{Pr}[K<m]\\
&=\epsilon\sum_{j=0}^{m-1}(1-F_j)+F_{m-1}\label{eq:ex-alg}.
\end{align}

First suppose $m \le 1/\epsilon$.
Then the optimum selects all online vertices, so $\mathrm{OPT}=m\epsilon$.
By $\gamma$-robustness, $\mathbb{E}[A_m]\le \gamma m\epsilon$.
Combined with (\ref{eq:ex-alg}), for any $0\leq r\leq 1/\epsilon-1$,
$$
F_r \le\frac{(\gamma-1)(r+1)\epsilon}{1-\epsilon}+\frac{\epsilon}{1-\epsilon}\sum_{j=0}^{r-1}F_j.
$$
Applying the discrete Gr\"onwall inequality~\cite{pachpatte1973discrete}, for any $r\leq 1/\epsilon-1$, we have
$$
F_r\le(\gamma-1)\bigl((1-\epsilon)^{-(r+1)}-1\bigr)=(\gamma-1)\bigl(e^{(r+1)\epsilon}-1+O(\epsilon)\bigr).
$$

Next consider prefixes with $m>1/\epsilon$.
In this case both the offline optimum and the advice-induced vertex cover have cost $1$.
We first prove that $\lim_{m\to\infty}F_m=1$.

Suppose, to the contrary, that this is not the case.
Then there exists $\epsilon_0>0$ such that, for every $N\ge 0$, there exists some $r\ge N$ satisfying $\mathrm{Pr}[K>r]>\epsilon_0$.
Let $N>\max\{1/\epsilon,\gamma/(\epsilon\epsilon_0)\}$.
Since $A$ is $\gamma$-robust, we have
$$
\gamma \ge \mathbb{E}[A_r] \ge \epsilon N\mathrm{Pr}[K>r]> \epsilon \epsilon_0 N> \gamma,
$$
which is a contradiction.
Therefore, $\lim_{m\to\infty}F_m=1$.
Let
$$
s=\log\left(1+\frac{1}{\gamma-1}\right),
\qquad
R=\lfloor s/\epsilon\rfloor-1.
$$
We first note that $s\le 1$.
Taking $m\to\infty$ in~(\ref{eq:ex-alg}) along prefixes with $m>1/\epsilon$, gives
$$
\gamma\geq 1+\epsilon\sum_{j=0}^{\infty}(1-F_j)\geq 1+\epsilon\sum_{j=0}^{R}(1-F_j).
$$
To justify $s\leq 1$, take $R=\lfloor1/\epsilon\rfloor-1$ in the preceding lower bound.
Using $\gamma$-robustness and letting $\epsilon\to 0$, we obtain
$$
\gamma\ge 1+\int_0^1\bigl(1-(\gamma-1)(e^x-1)\bigr)\,dx=2-(\gamma-1)(e-2),
$$
which implies $\gamma\ge e/(e-1)$ and therefore $s\le1$.
Thus $s\leq 1$ and $R\le 1/\epsilon-1$, and the bound on $F_j$ applies.
Therefore, as $\epsilon\to0$,
\begin{align*}
\liminf_{\epsilon\to 0}\liminf_{m\to\infty}\mathbb{E}[A_m]&\ge 1+\int_0^s \bigl(1-(\gamma-1)(e^x-1)\bigr)\,dx\\
&=\gamma s=\gamma\log\left(1+\frac{1}{\gamma-1}\right).
\end{align*}
Since the advice cost is $1$, any consistency ratio $c$ satisfies
$
c\ge\gamma\log\left(1+\frac{1}{\gamma-1}\right)$, and the theorem follows.

\paragraph{Proof of Theorem 4}
Let $A$ be a deterministic $(1+1/\lambda)$-robust algorithm.
Fix $\epsilon \in (0,\lambda/(\lambda+1))$.
Consider a star with $U=\{u\}$ and online vertices $V=\{v_1,\ldots,v_m\}$ where $w_u=1$ and every online vertex has weight $\epsilon$.
Every edge $(u,v_i)$ is revealed when $v_i$ arrives.
The advice is $a_v=0$ for all $v\in V$, so the advice-induced vertex cover is $\{u\}$.

We first note that $A$ cannot select $u$ before or upon the arrival of the first edge.
On the instance containing only $u$ and $v_1$, selecting $u$ gives $\ALG\ge 1$, whereas $\OPT=\epsilon$, contradicting
$(1+1/\lambda)$-robustness since
$$
1 > \Bigl(1+\frac1\lambda\Bigr)\epsilon .
$$
Thus $A$ must select $v_1$.

On the infinite sequence, $A$ must eventually select $u$.
Otherwise, if $A$ never selects $u$, then for every $m$, on the length-$m$ prefix it selects all $m$ leaves and pays $m\epsilon$.
Since $\OPT$ is $1$ for $m>1/\epsilon$, choosing $m>(1+1/\lambda)/\epsilon$ violates robustness.
Let $K$ be the number of leaves selected before $A$ first selects $u$.
Since $A$ must select $v_1$, we have $K\ge 1$.
We claim that
$$
K\epsilon \ge \lambda-(\lambda+1)\epsilon .
$$

If $(K+1)\epsilon\le 1$, consider the prefix ending at $v_{K+1}$. On this
prefix, $A$ incurs at least $1+K\epsilon$, while the optimum selects all
$K+1$ leaves and has cost $(K+1)\epsilon$. Robustness gives
\[
1+K\epsilon
\le
\Bigl(1+\frac1\lambda\Bigr)(K+1)\epsilon,
\]
which rearranges to the claimed bound. If instead $(K+1)\epsilon>1$, then
$K\epsilon>1-\epsilon\ge \lambda-(\lambda+1)\epsilon$, since $\lambda<1$.

Finally, take a finite prefix with $m>\max\{K,1/\epsilon\}$. Both the offline
optimum and the advice-induced vertex cover select only $u$, so $\OPT=\ADVICE=1$.
However, before selecting $u$, the algorithm has already selected $K$ leaves,
and hence
\[
c\geq \ALG/\ADVICE\ge 1+K\epsilon
\ge 1+\lambda-(\lambda+1)\epsilon .
\]
Letting $\epsilon\to 0$ yields $c\ge 1+\lambda$.

\bibliography{aaai2027}

\let\maketitle\arxivOriginalMaketitle
\let\bibliography\arxivOriginalBibliography

\bibliography{aaai2027}

\end{document}